\documentclass[conference]{IEEEtran}
\IEEEoverridecommandlockouts
\usepackage{cite}
\usepackage{amsmath,amssymb,amsfonts}
\usepackage{graphicx}
\usepackage{textcomp}
\usepackage{xcolor}
\def\BibTeX{{\rm B\kern-.05em{\sc i\kern-.025em b}\kern-.08em
    T\kern-.1667em\lower.7ex\hbox{E}\kern-.125emX}}

\usepackage{amsthm,mathtools}
\usepackage{algorithm}
\usepackage{algorithmic}
\usepackage{booktabs}
\usepackage{array}
\usepackage{url}
\usepackage{caption}

\newtheorem{theorem}{Theorem}
\newtheorem{observation}{Observation}
\newtheorem{proposition}{Proposition}
\newtheorem{fact}{Fact}
\newtheorem{corollary}{Corollary}
\newtheorem{remark}{Remark}
\newtheorem{example}{Example}

\newcommand{\tr}{\operatorname{tr}}
\newcommand{\lammax}{\lambda_{\max}}
\newcommand{\Hs}{\mathcal{H}}
\newcommand{\U}{\mathcal{U}}
\newcommand{\Sset}{\mathsf{S}}
\newcommand{\M}{\mathsf{M}}
\newcommand{\R}{\mathbb{R}}
\newcommand{\C}{\mathbb{C}}
\newcommand{\E}{\mathbb{E}}
\newcommand{\Var}{\operatorname{Var}}

\begin{document}

\title{Spectral Minimax Direct Fidelity Estimation\\for Generic Target States
}

\author{\IEEEauthorblockN{Hyunho Cha and Jungwoo Lee}
\IEEEauthorblockA{\textit{NextQuantum and Department of Electrical and Computer Engineering} \\
\textit{Seoul National University}\\
Seoul, Republic of Korea \\
\texttt{\{ovalavo, junglee\}@snu.ac.kr}}
}

\maketitle


\begin{abstract}
Direct fidelity estimation benefits from tailoring measurements to a fixed target, but the operator-aware shadow importance sampling (OASIS) method optimizes an outcome-wise linear-program surrogate rather than the exact worst-case variance over physical states. We propose an exact spectral replacement for arbitrary target states under the same non-adaptive single-copy measurement model. Specifically, we characterize unbiased linear estimators by a single operator identity, determine the state-wise optimal sampling law for fixed reconstruction coefficients, and convert the exact minimax problem into a semidefinite program. The resulting offline design and online estimator are presented as an algorithm and implemented with local Pauli measurements. Numerical simulations under depolarizing noise demonstrate that our exact spectral optimization outperforms the OASIS surrogate in terms of estimation variance.
\end{abstract}


\section{Introduction}
Estimating the fidelity between an unknown quantum state and a fixed pure target is a central primitive in verification, benchmarking, calibration, and experimental validation \cite{huang2025certifying}. The direct-fidelity-estimation framework showed that one can estimate a target overlap far more economically than by full state tomography when the target structure is exploited explicitly \cite{flammia2011direct}. The classical-shadows protocol then established a scalable operator-prediction viewpoint in which randomized measurements are followed by target-dependent classical post-processing \cite{huang2020predicting, huang2022learning}. These developments motivated a sequence of target-aware fidelity estimators that optimize not only the post-processing rule but also the measurement distribution itself \cite{barbera2025sampling, cha2025efficient, cha2026operator}.

A recent example is \emph{operator-aware shadow importance sampling} (OASIS) introduced in \cite{cha2026operator}. The key idea of OASIS is to expand the target projector in an informationally overcomplete measurement family rather than in a minimally complete basis \cite{acharya2021shadow, innocenti2023shadow, cha2026operator}. Once the target is represented in an overcomplete family, the reconstruction coefficients are no longer unique, and this nonuniqueness can be exploited to reduce estimator variance \cite{caprotti2024optimizing}. For arbitrary targets, OASIS solves a linear program whose objective penalizes, setting by setting, the largest outcome coefficient magnitude. This construction is experimentally natural and already improves the grouped direct-fidelity-estimation baseline for random targets \cite{barbera2025sampling,cha2026operator}. However, the true adversary is not an isolated measurement outcome but a quantum state that induces an entire outcome distribution.

That observation changes the optimization problem. If a measurement setting contains one large coefficient attached to an outcome that adversarial states can rarely realize, then an outcome-wise supremum is too pessimistic. The statistically relevant quantity is instead the exact state-dependent second moment. Optimizing against states rather than against worst isolated outcomes leads to a genuinely different objective, namely the exact worst-case one-shot variance over the same estimator class used by OASIS.

We restrict attention to generic target states. As in the arbitrary-target OASIS-GT regime of \cite{cha2026operator}, our formulation is mainly useful for generic or randomly generated targets rather than highly structured families such as GHZ or W, because specialized grouping methods can exploit sparse Pauli support whereas an arbitrary-target overcomplete optimization may still assign weight to redundant settings.

The comparison with OASIS is clean in two senses. First, both methods use the same experimentally relevant local-Pauli measurement family. Second, both belong to the same class of unbiased linear estimators built from non-adaptive single-copy measurements. The improvement comes from replacing the offline surrogate objective by the exact minimax objective.

The contributions are as follows. First, we define the estimator class carefully and observe that unbiasedness is equivalent to one linear operator identity. Second, we derive the exact one-shot variance and the state-wise optimal measurement law for fixed reconstruction coefficients. This step explains precisely why the outcome-wise OASIS objective is only an upper bound. Third, we prove that the OASIS linear program is a relaxation and can be strictly loose already in the simplest one-qubit setting. Fourth, we prove an exact spectral minimax identity for the worst-case variance and convert it into an exact semidefinite program (SDP) by a perspective lift. Finally, we present the resulting algorithm in offline-online form and state the corresponding worst-case mean-squared-error bound for the multi-shot estimator.

\section{Measurement model and unbiased estimators}
Let $\Hs=(\C^2)^{\otimes n}$ be the $n$-qubit Hilbert space and let $d=2^n$. Fix a unit vector $|\psi\rangle\in\Hs$ and denote the target rank-one projector by
$
O=|\psi\rangle\langle\psi|.
$
For any density operator $\rho\in\Sset(\Hs):=\{\rho\succeq 0:\tr(\rho)=1\}$, the target fidelity is
$
F(\rho)=\tr(\rho O).
$
We write $A\succeq 0$ for positive semidefiniteness, $A\preceq B$ for $B-A\succeq 0$, and $\lammax(A)$ for the largest eigenvalue of a Hermitian operator $A$.

A measurement setting family is a finite collection $\{P_{u,b}\}$ indexed by settings $u\in\U$ and outcomes $b\in B_u$ such that, for each fixed $u$, the operators $P_{u,b}\succeq 0$ satisfy $\sum_{b\in B_u}P_{u,b}=I$. Thus each setting defines an ordinary positive operator-valued measure (POVM). In the experiments of this paper we use the local Pauli family. For $u=(u_1,\ldots,u_n)\in\{X,Y,Z\}^n$, let $H$ be the Hadamard matrix, let $S=\operatorname{diag}(1,i)$, set $U_Z=I$, $U_X=H$, and $U_Y=HS^\dagger$, define $U_u=U_{u_1}\otimes\cdots\otimes U_{u_n}$, and for $b\in\{0,1\}^n$ set
\begin{equation}
P_{u,b}=U_u^\dagger |b\rangle\langle b|U_u.
\label{eq:pauli-povm}
\end{equation}
Operationally, choosing setting $u$ means rotating into the local Pauli eigenbasis indicated by $u$ and then measuring in the computational basis.

The role of overcompleteness is important, so we state it explicitly.

\begin{fact}
\label{fact:ioc}
Let $p$ be any probability distribution on $\U=\{X,Y,Z\}^n$ with $p(u)>0$ for every $u$. Then the weighted effects
$
E_{u,b}=p(u)P_{u,b}
$
form an informationally overcomplete POVM on $\Hs$ \cite{d2007optimal}.
\end{fact}

\begin{proof}
Each $E_{u,b}$ is positive semidefinite. Also,
\begin{equation}
\sum_{u\in\U}\sum_{b\in\{0,1\}^n}E_{u,b}=
\sum_{u\in\U}p(u)\sum_b P_{u,b}=
\sum_{u\in\U}p(u)I=I,
\end{equation}
so the family is a POVM. For one qubit, the six Pauli eigenprojectors are $(I\pm X)/2$, $(I\pm Y)/2$, and $(I\pm Z)/2$, and these span the four-dimensional real vector space of Hermitian $2\times2$ matrices because one can recover $I$, $X$, $Y$, and $Z$ from linear combinations of them. Tensor products of spanning sets span tensor-product spaces, so the tensor products of these six one-qubit projectors span the full real Hermitian operator space on $(\C^2)^{\otimes n}$, whose dimension is $4^n$. Since each $E_{u,b}$ is a positive scalar multiple of $P_{u,b}$, the weighted effects span the same space. Finally, there are $6^n$ effects, and $6^n>4^n$ for all $n\ge 1$, so the family is linearly dependent over $\R$. Therefore it is informationally overcomplete.
\end{proof}

A one-shot estimator is specified by a sampling law $q$ on $\U$ and real outputs $g_u(b)$ for $b\in B_u$. A single shot samples $u\sim q$, measures $\rho$ with the POVM $\{P_{u,b}\}_{b\in B_u}$, observes $b$, and returns the scalar output $X=g_u(b)$. It is convenient to absorb the measurement law into the coefficient parametrization
\begin{equation}
\alpha_{u,b}=q(u)g_u(b).
\label{eq:alpha}
\end{equation}
If $q(u)=0$, then necessarily $\alpha_{u,b}=0$ for every $b\in B_u$, because a setting that is never sampled cannot carry nonzero estimator weight.

The following criterion characterizes the entire unbiased estimator class.

\begin{observation}
\label{obs:unbiased}
The one-shot estimator defined by $(q,g)$ is unbiased for the target $O$ if and only if
\begin{equation}
\sum_{u\in\U}\sum_{b\in B_u}\alpha_{u,b}P_{u,b}=O.
\label{eq:unbiased}
\end{equation}
\end{observation}

\begin{proof}
Using Eq.~\eqref{eq:alpha},
\begin{align}
\E_\rho[X]
&=\sum_{u\in\U}q(u)\sum_{b\in B_u}\tr(\rho P_{u,b})g_u(b)\notag\\
&=\sum_{u,b}\tr(\rho P_{u,b})\alpha_{u,b}
 =\tr\!\left(\rho\sum_{u,b}\alpha_{u,b}P_{u,b}\right).
\label{eq:unbiased-proof}
\end{align}
Hence $\E_\rho[X]=\tr(\rho O)$ for every $\rho\in\Sset(\Hs)$ if and only if Eq.~\eqref{eq:unbiased} holds.
\end{proof}

Equation \eqref{eq:unbiased} is the exact connection with OASIS. If OASIS chooses a default distribution $p$ and expands the target as $O=\sum_{u,b}\omega_{u,b}E_{u,b}$ with $E_{u,b}=p(u)P_{u,b}$, then the present coefficients are simply $\alpha_{u,b}=p(u)\omega_{u,b}$ \cite{cha2026operator}. Thus the spectral formulation does not enlarge the feasible estimator class. It changes only the offline objective used to choose the same type of unbiased coefficients.

For an admissible pair $(q,\alpha)$ satisfying Eq.~\eqref{eq:unbiased}, define the second-moment operator
\begin{equation}
\M_{q,\alpha}=\sum_{u\in\U}\sum_{b\in B_u}\frac{\alpha_{u,b}^2}{q(u)}P_{u,b},
\label{eq:M}
\end{equation}
where terms with $q(u)=0$ are omitted because they necessarily have $\alpha_{u,b}=0$.

\begin{observation}
\label{obs:variance}
For every $\rho\in\Sset(\Hs)$,
\begin{equation}
\Var_\rho(X)=\tr(\rho\M_{q,\alpha})-\tr(\rho O)^2.
\label{eq:varformula}
\end{equation}
\end{observation}

\begin{proof}
On outcome $(u,b)$, the estimator outputs $X=\alpha_{u,b}/q(u)$. Therefore
\begin{align}
\E_\rho[X^2]
&=\sum_u q(u)\sum_b \tr(\rho P_{u,b})\left(\frac{\alpha_{u,b}}{q(u)}\right)^2\notag\\
&=\tr\!\left(\rho\sum_{u,b}\frac{\alpha_{u,b}^2}{q(u)}P_{u,b}\right)
 =\tr(\rho\M_{q,\alpha}).
\end{align}
Observation~\ref{obs:unbiased} gives $\E_\rho[X]=\tr(\rho O)$, and subtracting the square of the mean yields Eq.~\eqref{eq:varformula}.
\end{proof}

\section{From exact variance to exact design}
Observation~\ref{obs:variance} shows that the design problem is to minimize the worst-case value of the state-dependent variance $\Var_\rho(X)$ over $\rho\in\Sset(\Hs)$, subject to the unbiasedness identity. To understand why this differs from the OASIS surrogate, it is useful to optimize the sampling law for fixed reconstruction coefficients and a fixed state.

\begin{theorem}[State-wise optimal law for fixed coefficients]
\label{thm:fixedstate}
Fix coefficients $\alpha_{u,b}$ satisfying Eq.~\eqref{eq:unbiased} and fix a state $\rho\in\Sset(\Hs)$. Define
$
\mu_u(\rho)=\sum_{b\in B_u}\alpha_{u,b}^2\tr(\rho P_{u,b})\ge 0.
$
Then
\begin{equation}
\inf_{q}\sum_{u\in\U}\frac{\mu_u(\rho)}{q(u)}
=\left(\sum_{u\in\U}\sqrt{\mu_u(\rho)}\right)^2,
\label{eq:sqrt-rule-value}
\end{equation}
where the infimum is over all probability distributions $q$ on $\U$. If not all $\mu_u(\rho)$ vanish, equality is attained exactly by
\begin{equation}
q_\rho^{\star}(u)=\frac{\sqrt{\mu_u(\rho)}}{\sum_{v\in\U}\sqrt{\mu_v(\rho)}}
\qquad \text{whenever }\mu_u(\rho)>0.
\label{eq:sqrt-rule}
\end{equation}
\end{theorem}

\begin{proof}
If $q(u)=0$ for some index with $\mu_u(\rho)>0$, then the objective is infinite, so an optimal law must satisfy $q(u)>0$ on the support of $\mu_u(\rho)$. Apply the Cauchy--Schwarz inequality to the vectors with components $a_u=\sqrt{\mu_u(\rho)/q(u)}$ and $b_u=\sqrt{q(u)}$. Since $\sum_u q(u)=1$,
\begin{align}
\left(\sum_u \sqrt{\mu_u(\rho)}\right)^2
&=\left(\sum_u a_ub_u\right)^2 \notag\\
&\le \left(\sum_u \frac{\mu_u(\rho)}{q(u)}\right)
\left(\sum_u q(u)\right) \notag\\
&=\sum_u \frac{\mu_u(\rho)}{q(u)}.
\end{align}
This proves the lower bound in Eq.~\eqref{eq:sqrt-rule-value}. Equality in Cauchy--Schwarz holds precisely when the vectors $(a_u)_u$ and $(b_u)_u$ are proportional on the support of $\mu_u(\rho)$, which means $q(u)\propto\sqrt{\mu_u(\rho)}$. Normalization gives Eq.~\eqref{eq:sqrt-rule}, and substituting it into the objective yields Eq.~\eqref{eq:sqrt-rule-value}.
\end{proof}

We remark that Theorem~\ref{thm:fixedstate} should be interpreted as a structural result, not as the final arbitrary-target algorithm. In the actual estimation problem, the state $\rho$ is unknown, so one cannot deploy the state-dependent law $q_\rho^{\star}$ directly. What the theorem provides is the correct local geometry of the risk. Specifically, variance is governed by how an adversarial state distributes weight across the outcomes inside each setting, and those weights enter through quadratic averages rather than through a single largest coefficient. This is exactly why an outcome-wise $\ell_\infty$ surrogate can be loose. The theorem also clarifies the role of overcompleteness. Because the coefficients $\alpha_{u,b}$ are not unique, one can reshape the setting-wise second moments $\mu_u(\rho)$ without changing unbiasedness. The spectral minimax formulation in the next section performs that reshaping globally, simultaneously over all states, by replacing the unattainable family of state-dependent laws $q_\rho^{\star}$ with one state-independent design that is optimal in the exact minimax sense.

Theorem~\ref{thm:fixedstate} gives the main conceptual distinction between the exact problem and the OASIS surrogate. The exact state-dependent objective aggregates each setting through the square root of a \emph{state-weighted} second moment. By contrast, the OASIS linear program replaces this quantity by the much coarser setting-wise bound $\max_b|\alpha_{u,b}|$. The next proposition makes the relation precise.

\begin{proposition}[The OASIS linear program is a relaxation]
\label{prop:relax}
Define
\begin{equation}
L(\alpha)=\sum_{u\in\U}\max_{b\in B_u}|\alpha_{u,b}|.
\label{eq:lpobj}
\end{equation}
For every coefficient family $\alpha$ satisfying Eq.~\eqref{eq:unbiased}, there exists a sampling law $q_\alpha$ such that
\begin{equation}
\sup_{\rho\in\Sset(\Hs)}\Var_\rho(X)\le L(\alpha)^2.
\label{eq:relax-bound}
\end{equation}
Consequently,
\begin{equation}
\Gamma^{\star}(O):=\inf_{(q,\alpha):\,\mathrm{Eq.}\,\eqref{eq:unbiased}}\sup_{\rho\in\Sset(\Hs)}\Var_\rho(X)
\le \inf_{\alpha:\,\mathrm{Eq.}\,\eqref{eq:unbiased}}L(\alpha)^2.
\label{eq:relaxation}
\end{equation}
\end{proposition}

\begin{proof}
Set $m_u=\max_{b\in B_u}|\alpha_{u,b}|$, let $L=\sum_u m_u$, and choose
\begin{equation}
q_\alpha(u)=\begin{cases}
 m_u/L, & m_u>0,\\
 0, & m_u=0.
\end{cases}
\end{equation}
For any state $\rho$,
\begin{align}
\tr(\rho\M_{q_\alpha,\alpha})
&=\sum_u\sum_b\frac{\alpha_{u,b}^2}{q_\alpha(u)}\tr(\rho P_{u,b})\notag\\
&\le \sum_u\frac{m_u^2}{q_\alpha(u)}\sum_b\tr(\rho P_{u,b})
 =\sum_u Lm_u=L^2,
\label{eq:Lbound}
\end{align}
because $\sum_b\tr(\rho P_{u,b})=1$ for each $u$. Observation~\ref{obs:variance} then gives
$
\Var_\rho(X)=\tr(\rho\M_{q_\alpha,\alpha})-\tr(\rho O)^2\le L^2.
$
Taking the supremum over $\rho$ proves Eq.~\eqref{eq:relax-bound}, and taking the infimum over feasible $\alpha$ proves Eq.~\eqref{eq:relaxation}.
\end{proof}

The relaxation can be strict even when there is only one qubit and one measurement setting.

\begin{example}
\label{ex:gap}
Suppose the measurement family consists only of the computational-basis projectors $P_0=|0\rangle\langle0|$ and $P_1=|1\rangle\langle1|$, and the target is $O=|0\rangle\langle0|$. Then the exact worst-case one-shot variance equals $1/4$, whereas the OASIS surrogate value equals $1$.
\end{example}

\begin{proof}
There is only one setting, so $q=1$. Unbiasedness forces $\alpha_0P_0+\alpha_1P_1=O$, hence $\alpha_0=1$ and $\alpha_1=0$. Therefore the estimator is the Bernoulli random variable that returns $1$ with probability $f=\tr(\rho O)$ and $0$ otherwise. Its variance is $f-f^2$, maximized at $f=1/2$, so the exact worst-case variance is $1/4$. On the other hand, Eq.~\eqref{eq:lpobj} gives $L(\alpha)=\max\{|1|,|0|\}=1$, and the surrogate value is therefore $1$.
\end{proof}

Example~\ref{ex:gap} is small but informative. It shows that the exact objective cannot in general be recovered from a setting-wise maximum coefficient. A minimax treatment must optimize against physical states directly. The next section proves that the resulting worst-case variance has a clean spectral form and can still be solved exactly as a convex program.

\section{Exact spectral minimax formulation and SDP}
We now prove the central identity. For admissible $(q,\alpha)$, define the one-shot worst-case risk
\begin{equation}
\Gamma(q,\alpha;O)=\sup_{\rho\in\Sset(\Hs)}\Var_\rho(X).
\label{eq:gamma-fixed}
\end{equation}

\begin{theorem}[Exact spectral minimax identity]
\label{thm:spectral}
For every admissible pair $(q,\alpha)$,
\begin{equation}
\Gamma(q,\alpha;O)=
\min_{t\in\R}\Bigl\{\lammax\bigl(\M_{q,\alpha}-2tO\bigr)+t^2\Bigr\}.
\label{eq:spectral}
\end{equation}
Consequently,
\begin{equation}
\Gamma^{\star}(O)=
\inf_{(q,\alpha):\,\mathrm{Eq.}\,\eqref{eq:unbiased}}
\min_{t\in\R}\Bigl\{\lammax\bigl(\M_{q,\alpha}-2tO\bigr)+t^2\Bigr\}.
\label{eq:gamma-spectral}
\end{equation}
\end{theorem}

\begin{proof}
Fix $(q,\alpha)$ and abbreviate $\M=\M_{q,\alpha}$. For $t\in\R$, define
\begin{equation}
\Phi(\rho,t)=\tr(\rho\M)-2t\tr(\rho O)+t^2.
\label{eq:phi}
\end{equation}
Since $O$ is a projector, $0\le \tr(\rho O)\le 1$ for every state $\rho$, and therefore
\begin{equation}
\inf_{t\in\R}\Phi(\rho,t)=\inf_{t\in[0,1]}\Phi(\rho,t)=\tr(\rho\M)-\tr(\rho O)^2.
\label{eq:phi-min}
\end{equation}
By Observation~\ref{obs:variance}, the left-hand side of Eq.~\eqref{eq:spectral} is thus $\sup_\rho\inf_{t\in[0,1]}\Phi(\rho,t)$. The state set $\Sset(\Hs)$ is compact and convex, the interval $[0,1]$ is compact and convex, and $\Phi$ is continuous, affine in $\rho$, and convex in $t$. Sion's minimax theorem \cite{Sion1958general} therefore yields
\begin{equation}
\sup_{\rho\in\Sset(\Hs)}\inf_{t\in[0,1]}\Phi(\rho,t)
=\inf_{t\in[0,1]}\sup_{\rho\in\Sset(\Hs)}\Phi(\rho,t).
\label{eq:sion}
\end{equation}
For fixed $t$, maximizing the affine functional $\rho\mapsto\tr(\rho(\M-2tO))$ over density operators gives the largest eigenvalue, so
\begin{equation}
\sup_{\rho\in\Sset(\Hs)}\Phi(\rho,t)=\lammax(\M-2tO)+t^2.
\label{eq:lambda}
\end{equation}
Combining Eqs.~\eqref{eq:phi-min}, \eqref{eq:sion}, and \eqref{eq:lambda} shows that the minimum over $t\in[0,1]$ equals the worst-case variance.

It remains to justify that the minimum over $\R$ is attained in $[0,1]$. If $t<0$, then $\M-2tO=\M+2|t|O\succeq \M$, so
\begin{equation}
\lammax(\M-2tO)+t^2\ge \lammax(\M),
\end{equation}
which is the value at $t=0$. If $t>1$, then $O\preceq I$ implies $-2(t-1)O\succeq -2(t-1)I$, hence
\begin{align}
\M-2tO
&=(\M-2O)-2(t-1)O\notag\\
&\succeq (\M-2O)-2(t-1)I.
\end{align}
Therefore
\begin{align}
\lammax(\M-2tO)+t^2
&\ge \lammax(\M-2O)-2(t-1)+t^2\notag\\
&=\lammax(\M-2O)+1+(t-1)^2,
\end{align}
which is at least the value at $t=1$. Thus minimizing over $\R$ is equivalent to minimizing over $[0,1]$, proving Eq.~\eqref{eq:spectral}. Finally, taking the infimum over all admissible $(q,\alpha)$ yields Eq.~\eqref{eq:gamma-spectral}.
\end{proof}

Theorem~\ref{thm:spectral} is the exact arbitrary-target replacement for the OASIS linear program. It says that the entire fidelity-estimation design problem reduces to choosing a second-moment operator $\M_{q,\alpha}$ and then minimizing the rank-one-perturbed spectral quantity $\lammax(\M_{q,\alpha}-2tO)+t^2$.

We now convert the spectral formulation into an exact SDP. The standard perspective variable $y_{u,b}$ represents the ratio $\alpha_{u,b}^2/q(u)$, while a scalar slack variable $\gamma$ upper-bounds the largest eigenvalue term.

\begin{theorem}[Exact SDP]
\label{thm:sdp}
The value $\Gamma^{\star}(O)$ equals the optimum of
\begin{equation}
\begin{aligned}
\min_{\alpha,q,y,t,\gamma,s}\quad & \gamma+s\\
\mathrm{s.t.}\quad & \sum_{u,b}\alpha_{u,b}P_{u,b}=O,\\
& q(u)\ge 0\quad (u\in\U),\qquad \sum_{u\in\U}q(u)=1,\\
& \begin{bmatrix}y_{u,b} & \alpha_{u,b}\\ \alpha_{u,b} & q(u)\end{bmatrix}\succeq 0
\quad (u\in\U,\ b\in B_u),\\
& \sum_{u,b}y_{u,b}P_{u,b}-2tO\preceq \gamma I,\\
& \begin{bmatrix}s & t\\ t & 1\end{bmatrix}\succeq 0.
\end{aligned}
\label{eq:sdp}
\end{equation}
\end{theorem}

\begin{proof}
Take any admissible pair $(q,\alpha)$ and any $t\in\R$. Set $y_{u,b}=\alpha_{u,b}^2/q(u)$ when $q(u)>0$ and $y_{u,b}=0$ when $q(u)=0$, which is consistent because admissibility forces $\alpha_{u,b}=0$ whenever $q(u)=0$. Also set
\begin{equation}
\gamma=\lammax\!\left(\sum_{u,b}y_{u,b}P_{u,b}-2tO\right),
\qquad s=t^2.
\end{equation}
Then the unbiasedness identity and the probability constraints hold by construction. The block matrix
\begin{equation}
\begin{bmatrix}y_{u,b} & \alpha_{u,b}\\ \alpha_{u,b} & q(u)\end{bmatrix}
\end{equation}
is positive semidefinite because its determinant is zero when $q(u)>0$, and when $q(u)=0$ both entries in the off-diagonal vanish. The inequality
$
\sum_{u,b}y_{u,b}P_{u,b}-2tO\preceq \gamma I
$
holds by the definition of $\gamma$, while the final $2\times2$ block is positive semidefinite because $s=t^2$. Hence every spectral-feasible point yields an SDP-feasible point with the same objective value $\gamma+s=\lammax(\M_{q,\alpha}-2tO)+t^2$.

Conversely, take any SDP-feasible point. Positive semidefiniteness of each block matrix implies $y_{u,b}\ge \alpha_{u,b}^2/q(u)$ whenever $q(u)>0$, and forces $\alpha_{u,b}=0$ when $q(u)=0$. Therefore
\begin{equation}
\M_{q,\alpha}=\sum_{u,b}\frac{\alpha_{u,b}^2}{q(u)}P_{u,b}
\preceq \sum_{u,b}y_{u,b}P_{u,b}.
\label{eq:dominateM}
\end{equation}
Subtracting $2tO$ preserves semidefinite order, so
\begin{equation}
\lammax(\M_{q,\alpha}-2tO)
\le \lammax\!\left(\sum_{u,b}y_{u,b}P_{u,b}-2tO\right)
\le \gamma.
\label{eq:gamma-upper}
\end{equation}
Similarly, the final $2\times2$ block implies $s\ge t^2$. Hence every SDP-feasible point satisfies
$
\gamma+s\ge \lammax(\M_{q,\alpha}-2tO)+t^2,
$
which upper-bounds the exact spectral objective in Theorem~\ref{thm:spectral}. The two directions match, so Eq.~\eqref{eq:sdp} is exact.
\end{proof}

Algorithm~\ref{alg:spectral} states the resulting procedure. The offline stage solves the exact SDP, and the online stage performs the single-copy randomized measurements. In particular, the improved estimator does not require a different laboratory primitive. It only changes the target-dependent optimization that is performed before data collection. Any improvement will arise from a better choice of $q$ and $\alpha$, not from stronger measurements, additional adaptivity, or a change in the estimator class itself.

\begin{algorithm}
\caption{Spectral Minimax Direct Fidelity Estimation}
\label{alg:spectral}
\begin{algorithmic}[1]
\REQUIRE Target projector $O$, local-Pauli effects $\{P_{u,b}\}$ from Eq.~\eqref{eq:pauli-povm}, shot budget $N$
\STATE Solve the exact SDP in Eq.~\eqref{eq:sdp} offline and obtain an optimal law $q^{\star}$ and coefficients $\alpha^{\star}_{u,b}$.
\FOR{$j=1$ to $N$}
\STATE Sample $u_j\sim q^{\star}$.
\STATE Measure one copy of $\rho$ with $\{P_{u_j,b}\}_{b\in B_{u_j}}$ and get $b_j$.
\STATE Set $X_j=\alpha^{\star}_{u_j,b_j}/q^{\star}(u_j)$.
\ENDFOR
\STATE Return $\widehat F_N=\frac{1}{N}\sum_{j=1}^{N}X_j$.
\end{algorithmic}
\end{algorithm}

\begin{corollary}
\label{cor:valid}
For every $\rho\in\Sset(\Hs)$, Algorithm~\ref{alg:spectral} is unbiased and satisfies
\begin{equation}
\E_\rho[\widehat F_N]=F(\rho),
\quad
\sup_{\rho\in\Sset(\Hs)}\E_\rho\!\left[(\widehat F_N-F(\rho))^2\right]
\le \frac{\Gamma^{\star}(O)}{N}.
\label{eq:msebound}
\end{equation}
\end{corollary}

\begin{proof}
Each shot uses the same unbiased coefficient family from Observation~\ref{obs:unbiased}, so the average is unbiased. By Theorems~\ref{thm:spectral} and \ref{thm:sdp}, the one-shot variance is at most $\Gamma^{\star}(O)$. Independence of the $N$ shots divides the variance of the sample mean by $N$.
\end{proof}

\begin{proposition}[Operational meaning of the exact SDP output]
\label{prop:implementation}
Let $(\alpha^{\star},q^{\star},y^{\star},t^{\star},\gamma^{\star},s^{\star})$ be any optimal solution of Eq.~\eqref{eq:sdp}. Then the estimator executed during data collection depends only on the pair $(q^{\star},\alpha^{\star})$: on a sampled outcome $(u,b)$ it outputs the scalar $\alpha^{\star}_{u,b}/q^{\star}(u)$. Moreover, if $q^{\star}(u)=0$, then $\alpha^{\star}_{u,b}=0$ for every $b\in B_u$, so unsampled settings never require numerical outputs.
\end{proposition}

\begin{proof}
The online formula in Algorithm~\ref{alg:spectral} already uses only $q^{\star}$ and $\alpha^{\star}$. It remains to justify the zero-probability case. If $q^{\star}(u)=0$, positive semidefiniteness of the $2\times2$ block
\[
\begin{bmatrix}y^{\star}_{u,b} & \alpha^{\star}_{u,b}\\ \alpha^{\star}_{u,b} & q^{\star}(u)\end{bmatrix}
\succeq 0
\]
forces its off-diagonal entry to vanish, hence $\alpha^{\star}_{u,b}=0$. Therefore every sampled setting has a well-defined output coefficient and every unsampled setting carries zero estimator weight. The remaining variables $y^{\star}$, $t^{\star}$, $\gamma^{\star}$, and $s^{\star}$ certify exact optimality of the offline design through the linear matrix inequalities in Eq.~\eqref{eq:sdp}, but they are not needed at measurement time.
\end{proof}

Proposition~\ref{prop:implementation} is the precise implementation-level meaning of our improved method. In Section~\ref{sec:experiments}, the numerical comparison refers to the estimator obtained by solving the exact SDP in Theorem~\ref{thm:sdp} and then deploying the optimized pair $(q^{\star},\alpha^{\star})$ exactly as written in Algorithm~\ref{alg:spectral}.

\begin{corollary}[Additive-error sample complexity]
\label{cor:chebyshev}
For any $\varepsilon>0$ and $\delta\in(0,1)$, if
\begin{equation}
N\ge \frac{\Gamma^{\star}(O)}{\delta\,\varepsilon^2},
\label{eq:samplecomplexity}
\end{equation}
then Algorithm~\ref{alg:spectral} satisfies the uniform bound
\begin{equation}
\Pr_\rho\!\left\{|\widehat F_N-F(\rho)|\ge \varepsilon\right\}\le \delta
\qquad \text{for all }\rho\in\Sset(\Hs).
\label{eq:chebyshev}
\end{equation}
\end{corollary}

\begin{proof}
By Corollary~\ref{cor:valid},
\[
\sup_{\rho\in\Sset(\Hs)}\Var_\rho(\widehat F_N)\le \Gamma^{\star}(O)/N.
\]
Since $\widehat F_N$ is unbiased, Chebyshev's inequality gives
\[
\Pr_\rho\!\left\{|\widehat F_N-F(\rho)|\ge \varepsilon\right\}
\le \frac{\Var_\rho(\widehat F_N)}{\varepsilon^2}
\le \frac{\Gamma^{\star}(O)}{N\varepsilon^2}.
\]
The right-hand side is at most $\delta$ under Eq.~\eqref{eq:samplecomplexity}.
\end{proof}

Corollary~\ref{cor:chebyshev} shows that the exact spectral objective controls not only the worst-case $1/N$ MSE constant in Eq.~\eqref{eq:msebound} but also a nonasymptotic additive-error guarantee. Therefore, when the matched-budget experiments show a smaller MSE for Algorithm~\ref{alg:spectral}, they are probing a quantity that also governs a uniform concentration bound for the same optimized estimator.

\begin{remark}[Size of the exact local-Pauli SDP]
\label{remark:size}
For arbitrary-target local-Pauli measurements on $n$ qubits, the SDP in Eq.~\eqref{eq:sdp} contains $3^n$ sampling variables $q(u)$, $6^n$ coefficient variables $\alpha_{u,b}$, $6^n$ lifted variables $y_{u,b}$, $6^n+1$ small $2\times2$ positive-semidefinite constraints, and one Hermitian linear matrix inequality acting on a $2^n$-dimensional Hilbert space. In particular, the exact offline design problem for arbitrary targets is exponential in $n$.
\end{remark}

Remark~\ref{remark:size} explains the numerical regime of our algorithm. The arbitrary-target spectral formulation is exact, but it is not a scalable method for generic targets. Rather, it provides the correct offline benchmark within the measurement model, and for small $n$ it can be solved exactly and then reused across arbitrarily many experimental repetitions on the same target.

The improved method can still be operationally simple once the offline solve is completed. Although the optimization variables scale with the number of setting-outcome pairs, the deployed estimator consists only of a sampler for $q^{\star}$ and a lookup table for the outputs $\alpha^{\star}_{u,b}/q^{\star}(u)$. Thus the per-shot laboratory routine remains straightforward.

\section{Experimental results}
\label{sec:experiments}
For each $n\in\{3,4,5,6\}$, we draw a Haar-random pure target state \cite{haar1933massbegriff}, construct its projector $O$, and compare the OASIS estimator of \cite{cha2026operator} with the exact-SDP implementation of Algorithm~\ref{alg:spectral}. The measurement family is the same local Pauli family for both methods, so the comparison isolates the effect of replacing the OASIS surrogate objective by the exact spectral minimax objective.

The unknown state is the depolarized model
$
\rho=0.9\,O+0.1\,I/d.
$
Since $O$ is a rank-one projector, the corresponding exact target fidelity is
$
F(\rho)=\tr(\rho O)=0.9+0.1/d.
$
Hence no auxiliary tomography is needed to evaluate the estimation error. Specifically, both methods are evaluated against the same scalar ground truth $0.9+0.1/2^n$ and under the same total shot budget.

Because the arbitrary-target SDP assigns variables to fine-grained local-Pauli outcomes, its offline dimension still grows exponentially with $n$. The present numerical study is therefore limited to $n=3,4,5,6$, where exact offline optimization remains computationally manageable.

\begin{table}[h]
\captionsetup{justification=justified, singlelinecheck=false}
\caption{MSE comparison for Haar-random targets (in units of $10^{-4}$). Each entry is the average of $1000$ independent experiments with $\rho=0.9O+0.1I/d$. The number of shots was matched to that used by the grouping-based method of \cite{barbera2025sampling} with $\epsilon=\delta=0.1$.}
\label{tab:mse}
\centering
\footnotesize
\begin{tabular}{c>{\centering\arraybackslash}p{0.35\linewidth}>{\centering\arraybackslash}p{0.2\linewidth}>{\centering\arraybackslash}p{0.2\linewidth}}
\toprule
$n$ & \# Shots (matched to \cite{barbera2025sampling}) & OASIS \cite{cha2026operator} & Algorithm~\ref{alg:spectral} \\
\midrule
3 & 4426 & 3.67 & 3.49 \\
4 & 8127 & 2.40 & 2.13 \\
5 & 14083 & 2.22 & 1.52 \\
6 & 27399 & 1.49 & 1.01 \\
\bottomrule
\end{tabular}
\end{table}

For each configuration, the empirical mean-squared error is averaged over $1000$ independent experiments. Table~\ref{tab:mse} summarizes the MSE values.

To keep the comparison fair, both methods used the same total shot budget for a given $n$. Because the improved method changes only the offline optimization, this matched-budget protocol attributes any observed MSE reduction to the improved estimator design rather than to extra data. The relative MSE advantage of Algorithm~\ref{alg:spectral} becomes larger as $n$ increases.

A practical point is worth stressing. The exact spectral formulation is an offline design problem. Once $q^{\star}$ and $\alpha^{\star}_{u,b}$ are computed, the online estimator is operationally as simple as the OASIS estimator. Each shot samples a local Pauli setting, records a bit string, and outputs a scalar lookup value $\alpha^{\star}_{u,b}/q^{\star}(u)$. Consequently, the experimental burden remains in the same regime as OASIS.

\section{Conclusion}
For arbitrary target states, the exact adversary is a quantum state, not a worst isolated outcome inside each measurement setting. Once that distinction is enforced, the design problem becomes the exact minimax optimization of one-shot variance over the same non-adaptive single-copy estimator class used in \cite{cha2026operator}. In this work we presented a derivation of that exact replacement.

Algorithm~\ref{alg:spectral} is mathematically exact, solving the true worst-case variance instead of a surrogate. By maintaining the OASIS protocol's experimental structure, it proves that performance gains arise from an improved offline objective rather than superior measurement primitives.

An important direction for future work is to preserve the exact minimax viewpoint while reducing the offline complexity of the design problem. In particular, identifying target families with symmetry, tensor-network descriptions, or other compressed representations could allow the exact spectral SDP to admit smaller formulations or analytically simplified solutions.

\section*{Acknowledgment}
This work is in part supported by the National Research Foundation of Korea (NRF, RS-2024-00451435 (20\%), RS-2024-00413957 (20\%)), Institute of Information \& communications Technology Planning \& Evaluation (IITP, RS-2025-02305453 (15\%), RS-2025-02273157 (15\%), RS-2025-25442149 (15\%), RS-2021-II211343 (15\%)) grant funded by the Ministry of Science and ICT (MSIT), Institute of New Media and Communications (INMAC), and the BK21 FOUR program of the Education, Artificial Intelligence Graduate School Program (Seoul National University), and Research Program for Future ICT Pioneers, Seoul National University in 2026.

\bibliographystyle{IEEEtran}
\bibliography{references}

\end{document}